\documentclass[12pt, oneside]{amsart}   	
\usepackage[english]{babel}
\usepackage[utf8x]{inputenc}
\usepackage[T1]{fontenc}

\usepackage[a4paper,top=3cm,bottom=3cm,left=3cm,right=3cm,marginparwidth=1.75cm]{geometry}

\usepackage{comment}
\usepackage{bbm}
\usepackage{amsmath}
\usepackage{graphicx}
\usepackage[colorinlistoftodos]{todonotes}
\usepackage[colorlinks=true, allcolors=blue]{hyperref}
\usepackage{amsthm}
\usepackage{amssymb}
\usepackage{epsfig}
\usepackage{graphicx, amsmath, bm}
\usepackage{tabu}
\usepackage{tikz}
\usepackage{float}
\usepackage{booktabs}

\usetikzlibrary{graphs}
\usetikzlibrary{graphs.standard}
\usepackage{spurpaper} 
\DeclareMathOperator*{\expect}{\mathbb{E}}
\DeclareMathOperator{\Inf}{Inf}
\DeclareMathOperator{\I}{\mathbf{I}}
\DeclareMathOperator{\MAJ}{MAJ}
\DeclareMathOperator{\sgn}{sgn}

\newcommand{\quadgraph}[1]
{\tikzstyle{every node}=[circle, draw, fill=black!50, inner sep=0pt, minimum width=3pt]
\begin{tikzpicture}[thick,baseline=-4pt]
\node at (-0.25,-0.25)(1){};
\node at (-0.25, 0.25)(2){};
\node at ( 0.25, 0.25)(3){};
\node at ( 0.2, -0.25)(4){};
\draw { #1 };
\end{tikzpicture}}

\title[On Graphs and the Gotsman-Linial Conjecture for $d=2$]{On Graphs and the Gotsman-Linial Conjecture for $d=2$}
\author{Hyo Won Kim, Christopher Maldonado and Jake Wellens} 
\begin{document}

\begin{abstract}
Given a polynomial $p(x) : \{-1,1\}^n \to \R^{\times}$, the associated polynomial threshold function (PTF) is a boolean function $f : \{-1,1\}^n \to \{-1,1\}$ defined by $f(x) = \sgn(p(x))$. A conjecture of Gotsman and Linial \cite{GL} posits that among all degree-$d$ PTFs, the one with the largest influence corresponds to the symmetric degree-$d$ polynomial which alternates sign at the $d+1$ values of $\sum_{i=1}^n x_i$ closest to 0, having influence $\Theta(d \sqrt{n})$. We give an infinite class of counterexamples when $d = 2$, thereby disproving the conjecture as originally stated. However, as a theorem of Kane \cite{Kane} shows, for constant $d$, and any degree-$d$ PTF $f$ on $n$ variables, $\I[f] = O(\sqrt{n}\cdot\poly\log(n))$, so at least the conjectured bound $O(d\sqrt{n})$ is not too far off for small $d$. We examine the case $d=2$, i.e. when $f(x) = \sgn(x^TAx + b^Tx + c)$, and using only elementary methods, we remove the $\poly\log(n)$ from Kane's bound in a variety of special cases, based on graph properties of the matrix $A$, interpreted as a weighted adjacency matrix. \end{abstract}
\maketitle

\section{Introduction}

An useful way to measure the complexity of a boolean function $f: \{-1,1\}^n \to \{-1,1\}$ is through its \emph{influence} or \emph{average sensitivity}, defined as
$$\I[f] := \sum_{i=1}^n \Pr_{x \sim \{-1,1\}^n}[f(x_1, \dots, x_i, \dots, x_n) \neq f(x_1, \dots, -x_i, \dots, x_n)].$$
Identifying boolean functions with subsets of the hypercube graph on $\{-1,1\}^n$, influence counts (up to a factor of $2^{-n+1}$) the number of \emph{boundary edges} of such a subset. By way of Fourier analysis, influence can be related to many other complexity measures, such as noise sensitivity\footnote{The \emph{noise sensitivity} $\NS_{\delta}(f)$ of $f: \{-1,1\}^n \to \{-1,1\}$ is defined as $\Pr_{x \sim \{-1,1\}^n}[f(x) \neq f(y)]$, where each bit $y_i$ of $y$ is formed independently by flipping $x_i$ with probability $\delta$.}, spectral concentration, decision tree and circuit complexity -- see \cite{ABF} for a largely self-contained summary of these relationships. To be imprecise, the general trend is that \emph{low-influence functions are ``easier to compute''} in many models of computation than are those with high influence. Thus, lower bounds on influence can provide circuit lower bounds (e.g. $\PARITY \not \in \AC^0$), while upper bounds on influence can provide algorithms for PAC learning classes of functions from random examples. This paper is concerned with upper bounding the influence of a class of boolean functions called $\emph{quadratic threshold functions}$.

\emph{Polynomial threshold functions} (PTFs), more generally, are boolean functions $f: \{-1,1\}^n \to \{-1,1\}$ which can be represented as the sign of a polynomial, i.e.
$$f(x) = \sgn(p(x_1, \dots, x_n))$$
for some multilinear polynomial $p$.\footnote{For obvious reasons, such functions are also called \emph{threshold of parities circuits}. Note also that any boolean function on n variables is equal to the sign of its Fourier transform, and is therefore trivially a degree-$n$ PTF. Usually one is concerned with functions having \emph{low-degree} PTF representations, often much lower than the degree of the function's Fourier transform.} When the degree of $p$ is equal to $d$, we say $f$ is a degree-$d$ PTF. Perhaps most common in applications are PTFs of degree $d=1$ or $d=2$, known as \emph{linear threshold functions} (LTFs) and \emph{quadratic threshold functions} (QTFs), respectively.  LTFs have been extensively studied in the literature dating back to the 1960's, as such functions play a vital role in computational models of neural networks (see \cite{MP} for a classical introduction). In particular, it has long been known (as early as 1971, see \cite{PO}) that, among linear threshold functions, the majority function $$\MAJ_n(x_1, \dots, x_n) = \sgn(x_1 + \cdots + x_n)$$
has the largest total influence of $n$-variable LTFs. This influence can be computed explicitly and shown to be $\Theta(\sqrt{n})$. This fact has been shown (\cite{Harsha},\cite{Peres}) to imply a noise-sensitivity bound of $\NS_{\delta}(f) \leq O(\sqrt{\delta})$ for any LTF $f$, which in turn implies $f$ is $\epsilon$-concentrated up to degree $O(1/\epsilon^2)$, and therefore $f$ can be PAC learned\footnote{This is not the first or the fastest algorithm for learning halfspaces. However, since $\NS_{\delta}(f_1 \land f_2) \leq \NS_{\delta}(f_1) + \NS_{\delta}(f_2)$, this also yields the first algorithm for PAC learning \emph{intersections} of O(1) halfspaces in polynomial time.} from random examples in time $n^{O(1/\epsilon^2)}$.

Like their linear counterparts, quadratic threshold functions show up in a variety of real-world applications, such as Boltzmann machines (a type of stochastic neural network) and the Ising model (for describing the statistics of ferromagnetic spin configurations), although much less is known about QTFs than LTFs. We do not know which QTF has the largest influence, nor do we know tight bounds on influence, noise stability or spectral concentration.

Noting that $\MAJ_n$ is the unique unbiased LTF which is \emph{symmetric} in its $n$ variables, one might expect something similar to hold for PTFs of any degree $d$ --- namely that, among all degree $d$ PTFs $f$ on $n$ variables, the maximum value of $\I[f]$ is obtained by a symmetric function. In that case, the maximizer would be of the form
$$f(x) = \sgn\left(p_d\left(\sum_{i=1}^n x_i\right)\right)$$
for some univariate degree-$d$ polynomial $p_d(t)$. In this case, $f$ would be sensitive to $x$ in the $i$th coordinate if changing the sign of $x_i$ shifts $\sum_{i=1}^n x_i$ across a root of $p_d(t)$, and since $\Pr_{x \sim \{-1,1\}^n}[ \sum_{i=1}^n x_i = k]$ decreases as $k$ moves away from 0, the $p$ which maximizes $\I[f]$ will be one which alternates sign at the $d+1$ values of $\sum_{i=1}^nx_i$ closest to 0. Explicitly, if $n$ is odd, then one may take
$$p_d(t) := t(t-2)(t+2)(t-4)(t+4)\cdots (t+(-1)^{\lfloor d/2 \rfloor}\cdot 2 \lfloor d/2 \rfloor),$$
while for even $n$, one may take $p_d(t)$ to be the same as above with $t+1$ plugged in for $t$. Indeed, this reasoning led Gotsman and Linial to make the following conjecture: 

\begin{conj*}[Gotsman-Linial \cite{GL}]: The degree-$d$ PTF on $n$ variables with maximal influence is fully symmetric. Equivalently, for any PTF $f$ on $n$ variables with degree $d\ge2$, \footnote{Note that the conjecture is usually stated using only the form given here for even $n$, with no mention of the parity of $n$. It is shown in Appendix \ref{igl} to this paper that this is an error, perhaps propagating from the originally ambiguous floor/ceiling notation in \cite{GL}.}
$$\I[f] \leq \I[\sgn(p_d(x_1 + \cdots + x_n))] =  \begin{cases} 2^{-n+1}\displaystyle\sum_{k=0}^{d-1}\dbinom{n}{\lfloor (n-k)/2 \rfloor}(n- \lfloor (n-k)/2 \rfloor), & n \text{ even} \\ 2^{-n+1}\displaystyle\sum_{k=0}^{d-1}\dbinom{n}{\lceil (n-k)/2 \rceil}(n- \lceil (n-k)/2 \rceil), & n \text{ odd}

\end{cases}$$
$$= : \I_{GL}(n,d).$$
\end{conj*}

We show in this paper that the Gotsman-Linial conjecture is not, in general, true.\footnote{In the midst of our research, we became aware of unpublished and independent work by Brynmor Chapman, who discovered counterexamples to many cases of the conjecture for $d \geq 2$.} Using a linear-programming based algorithm described in Section \ref{counterexample}, we were able to identify the following 3-variable quadratic polynomial, which produces counterexamples to the conjecture in the case of $d=2$ for all odd $n\ge 5$. 
\be \nn
p(x, y, z) := 2x(1 - 7y + z) + 4 y - 7y^2 + 4yz + 6z + 3z^2, 
\ee

Indeed:

\begin{thm} \label{generalize}
For odd $n\ge 5$, let 
$$f_n(x):=\sgn\left(p\left(x_1,\sum_{i=2}^{n-2} x_i, x_{n-1}+x_n\right)\right)$$

Then,
$$\frac{\I[f_n]}{\I_{GL}(2,n)} = 1 + \frac{1}{8n} + O\mathopen{}\left(\frac{1}{n^2}\right). $$

\end{thm}
It is shocking, in the authors' humble opinion, that the “influence maximizer is
symmetric” principle which holds for linear threshold functions (as well as for $d = n$) does not persist for
$d = 2$. However, for most applications, it would suffice to know only that the Gotsman-Linial bound is of the \emph{correct asymptotic order} as $n \to \infty$. It's not difficult to show that, for $d < \sqrt{n}$, $\I_{GL}(n,d)$ is of order $\Theta(d \sqrt{n})$, and hence a natural relaxation of the original Gotsman-Linial conjecture is the following:

\begin{conj*}[Weak Gotsman-Linial]
For any degree-$d$ threshold function $f$ on $n$ variables,
$$\I[f] \leq O(d\cdot\sqrt{n}).$$
\end{conj*}

Using arguments from \cite{Harsha} and \cite{ABF}, the weaker conjecture would yield all of the usual desired properties for low-degree PTFs: $O(d\sqrt{\delta})$ noise sensitivity, $O(d^2/\epsilon^2)$-concentration, and PAC learnability in time $n^{O(d^2/\epsilon^2)}$. To the best of our knowledge, the best known upper bound on influence for general PTFs is the following theorem of Kane \cite{Kane}:

\begin{thm*}[Kane, 2013] There exists a constant $M$ such that, for any degree-$d$ PTF $f$ on $n$ variables, 
$$\I[f] \leq \sqrt{n}\cdot (2^d\cdot \log{n})^{M\cdot d\log d}.$$
\end{thm*}

For constant $d$, Kane's result comes close to the conjectured $O(\sqrt{n})$, differing only by a polylogarithmic factor. However, it quickly becomes trivial when $d$ is allowed to grow with $n$ (for $d = \Omega(\sqrt{\log n/\log \log n})$ the bound is already $\Omega(n)$), and even for constant $d$, it still leaves open the question of whether a polylogarithmic factor is necessary: can an asymmetric QTF, for instance, really have influence which is $\omega(1)$-times bigger than all symmetric QTFs?

In this paper, we answer this question for a few important classes of QTFs. We define the notion of \emph{support} for a QTF
$$f(x_1, \dots, x_n) = \sgn\left(\sum_{1 \leq i < j \leq n} a_{ij}x_ix_j + \sum_{i = 1}^n b_ix + c\right)$$
as the graph\footnote{More precisely, since a boolean function may have many representations as a QTF, we say $f$ is supported on $G = (V,E)$ if it \emph{has a representation} with $\{(i,j) : a_{ij} \neq 0\} \subseteq E$.} $G = (V, E)$ on $n$ vertices such that $E = \{(i,j) : a_{ij} \neq 0\}$. Almost all pre-existing work on bounding the influence of threshold functions follows a similar paradigm: estimates are first derived in the Gaussian setting for ``regular'' PTFs (i.e. those for which no single coordinate has too much influence) using anti-concentration results for low-degree polynomials, then transferred to the Bernoulli setting by way of an invariance principle, and extended to ``irregular'' PTFs using a regularity lemma. (This approach was used in \cite{Harsha} to obtain the first non-trivial upper bound of $O\mathopen{}\left(n^{1 - \frac{1}{4d+2}}\right)$ for the influence of degree-$d$ PTFs, and subsequently improved by Kane in \cite{Kane}.) Our approach uses none of this machinery --- in fact, our methods are elementary and completely independent of any information about the coefficients other than which ones are non-zero.\footnote{Of course, one can always perturb coefficients slightly without changing a PTF's truth table, so this is still in some sense an ``open'' condition.} Using a simple \emph{induced subgraph covering lemma}, proved in Section \ref{bounds}, we are able to obtain the following:

\begin{thm} \label{fracch}
If $f:\set{-1,1}^{n}\rightarrow\set{-1,1}$ is a QTF supported on a graph $G$ with fractional chromatic number\footnote{This is a standard term from graph theory, and we give the precise definition in Section \ref{bounds}. For now, the unfamiliar reader may replace $\chi_f(G)$ by the chromatic number $\chi(G)$, since $\chi_f(G) \leq \chi(G)$.} $\chi_f(G)$, then
$$\I[f] \leq \sqrt{\chi_f(G)}\cdot\sqrt{n}.$$
\end{thm}

From this theorem we can deduce that the Weak Gotsman-Linial conjecture  holds for QTFs supported on $O(1)$-partite graphs, planar graphs, graphs of bounded degree, to name a few. While such QTFs make up a ``small fraction'' of the space of all QTFs, they do encompass many of the functions that actually appear in practice. Restricted Boltzmann machines (RBMs), for example, commonly used in deep learning networks, are supported on bipartite graphs. The Ising model, too, is most often applied to graphs covered by this theorem: since particles are most often arranged in lattices, the graph of nearest-neighbor interactions is usually $O(1)$ colorable (depending only on the lattice structure and not the number of vertices/particles).

Using a slight modification of an argument from \cite{Harsha}, we are also able to prove a similar bound for sparse graphs:

\begin{thm} \label{edgebound}
If $f:\set{-1,1}^{n}\rightarrow\set{-1,1}$ is a QTF supported on a graph $G = (V, E)$, then
$$\I[f] \leq \sqrt{n + \sqrt{2|E|n}} .$$ 
\end{thm}

We also remark that our covering lemma gives a way of ``throwing out'' a number of problematic vertices. Indeed, using this observation, one can extend Theorems \ref{fracch} and \ref{edgebound} to a slightly wider class of graphs: if $f$ is supported on a graph $G'$ which satisfies $G' \setminus H = G$ for some induced subgraph $H$, then $f$ obeys the same influence bound as if it were supported on $G$, with an extra additive term 
$$\I[H] := \max_{g \in \text{QTFs supported on H}} \I[g].$$
In particular, combining this observation with Kane's $\sqrt{n} \cdot \poly\log(n)$ bound, we conclude that the weak Gotsman-Linial conjecture holds for QTFs supported on graphs which, after removing $n/\Omega(\poly\log(n))$ vertices, have $O(1)$ fractional chromatic number or $O(n)$ edges. As an example, this applies to graphs with a small ($O(n/\poly\log(n))$) vertex cover.

We prove our main theorems in Section \ref{bounds}. In Section \ref{counterexample}, we discuss counterexamples to the Gotsman-Linial conjecture, and raise a few interesting questions regarding the effect of a QTF's support on its influence. 

\section{Influence bounds for QTFs supported on special graphs} \label{bounds}
 \subsection{Preliminaries}
 
\begin{defn} \label{inf_i}
The \emph{influence} of the $i$th coordinate on $f: \{-1,1\}^n \to \{-1,1\}$ is defined as
\begin{equation*}
\Inf_i[f] := \Pr_{\bm{x}\sim\set{-1,1}^n}[f(\bm{x}) \neq f(\bm{x}^{\oplus i})],
\end{equation*} 
where $ x^{\oplus i} = (x_1,x_2,...,-x_i,x_{i+1},...,x_n).$ Moreover, the \emph{(total) influence} is
\begin{equation*}
\I[f] := \sum_{i=1}^n \mathrm{Inf}_{i}[f].
\end{equation*}
\end{defn}

\begin{defn}
We say that a QTF $f: \{-1,1\}^{n} \to \{-1,1\}$ of the form
\begin{equation*}
f(x)=\mathrm{sgn}\left(\sum_{i<j} a_{i,j} x_i x_j + \sum_{i=1}^n a_i x_i + a_0\right)
\end{equation*}
is \emph{supported} on the graph $G=(V,E)$, if $\set{(i,j): a_{i,j} \neq 0} \subset E.$ Moreover, let $Q_{G}$ denote the set of QTFs supported on $G$.
\end{defn}

\begin{defn}
Let $G=(V,E)$ be a graph on $n$ vertices. We define the \emph{maximal QTF influence} of $G$ to be 
$$ \I[G] := \max_{f \in Q_{G}} \I[f].$$
\end{defn}

\begin{defn} \label{derivative}
Let $f:\{-1,1\}^{n}\rightarrow \R$ be a real Boolean function. The $i$th discrete derivative operator maps $f$ to the function $D_{i}f:\{-1,1\}^{n-1}\rightarrow \R$ defined by 
$$ D_{i}f(x) = \frac{f(x^{i\rightarrow 1})-f(x^{i\rightarrow -1})}{2}, $$ where $ x^{i\rightarrow a} = (x_{1}, \cdots, x_{i-1}, a, x_{i+1}, \cdots, x_{n})$.
\end{defn}

For a fixed $x$, note that $D_{i}f(x) = \pm 1$ whenever $f(x) \neq f(x^{\oplus i})$, so we obtain the useful fact:

\begin{fact}For all boolean functions $f$,
$$\Inf_{i}[f] = \expect_{x \sim \set{-1,1}^{n}}[|D_{i}f(x)|]$$
\end{fact}

The following lemma is a restatement of the fact that among LTFs, majority maximizes the total influence. For the sake of completion, this result is proved in Appendix \ref{fourieranalysis}.
\begin{lem} \label{maj_G}
For an edgeless graph $G$ on $n$ vertices, we have
$$\I[G]=\I[\MAJ_n]=\sqrt{2/\pi}\sqrt{n}+O(1).$$
\end{lem}

\begin{defn}
Let $f: \set{-1,1}^n \to \mathbb{R}$, and $J \subset [n]$. Then $f_{J|z}: \set{-1,1}^{|J|} \to \mathbb{R}$ is the restriction of $f$ to $J$, acquired by fixing the coordinates in $[n]\setminus J$ to $z$.
\end{defn}
 
The following simple lemma provides the basis for many of the results used in the proofs of our main theorems.

\begin{lem}\label{infrestriction}
Let $f:\{-1,1\}^{n}\rightarrow \{-1,1\}$ be a Boolean function and $J\subset[n]$. Then, for $i \in J$, $$ \Inf_i[f] = \expect_{z \in \{-1,1\}^{|J^{c}|}}[\Inf_i[f_{J|z}]]. $$
\end{lem}
\begin{proof}
We have that
\begin{align*}
\Inf_i[f] &= \expect_{x \in \{-1,1\}^{n}}[D_{i}f^{2}] \\
&= \sum_{z \in \{-1,1\}^{|J^{c}|}} \frac{1}{2^{|J^{c}|}} \cdot \expect_{x \in \{-1,1\}^{|J|}}[D_{i}f_{J|z}^{2}] \\
&= \expect_{z \in \{-1,1\}^{|J^{c}|}}[\Inf_i[f_{J|z}]].
\end{align*}

\end{proof}

A boolean function's influence can also be defined in terms of its Fourier coefficients. A brief overview of Fourier analysis is included in Appendix \ref{fourieranalysis} for completion. 
\subsection{Proof of Main Theorems} 

\begin{lem}[Covering Lemma]
Let $G =(V,E)$ be a graph that admits a covering by induced graphs $G_{1}, \dots , G_{k}.$ Then,
$$ \I[G] \leq \I[G_{1}] + \dots + \I[G_{k}].$$
\end{lem}

\begin{proof}
Without loss of generality, assume that the graphs $G_{i}$ are disjoint. Let $f:\{-1,1\}^{n}\rightarrow \{-1,1\}$ be a QTF supported on $G$. By Lemma \ref{infrestriction},

$$ \sum_{j \in G_{i}}\Inf_{j}[f] = \expect_{z \sim \{-1,1\}^{|G_{i}^{c}|}}[\I[f_{G_{i}|z}]]. $$ 
Since $f_{G_{i}|z}$ is supported on $G_{i}$ for any $z \in \{-1,1\}^{|G_{i}^{c}|},$
$$ \expect_{z}[\I[f_{G_{i}|z}]] \leq \I[G_{i}], $$
and so 
$$ \I[f] = \sum_{i=1}^{k} \sum_{j \in G_{i}} \Inf_{j}[f] \leq  \I[G_{1}] + \dots + \I[G_{k}]. $$
\end{proof}

\begin{cor}
For any $f\in Q_G$ for n-vertex graphs G with chromatic number $\chi(G)$,
$$\I[f] \le \sqrt{\chi(G)}\sqrt{n}$$
in fact, we can also show 
$$\I[G] \le \sqrt{2/\pi}\sqrt{\chi(G)}\sqrt{n}+O(\chi(G))$$
\end{cor}
\begin{proof}
Let $G_1,...G_{\chi(G)}$ be the monochromatic sets formed when $\chi(G)$-coloring $G$. Each $G_i$ is an independent set, and thus $f$ is an LTF for all $f \in G_i$. Therefore, $\I[G_i] = \I[\MAJ_{|G_i|}]$. Then, from the covering lemma, we find

$$\I[G]\le \sum_{i=1}^{\chi(G)}\I[\MAJ_{|G_i|}] \le \sqrt{2/\pi}\sum_{i=1}^{\chi(G)}\sqrt{|G_i|} +O(\chi(G))$$

As $\sum_i G_i = n$, by Cauchy-Schwarz, we can show 
$$\sum_{i=1}^{\chi(G)} \sqrt{|G_i|}\le \sqrt{\chi(G)}\sqrt{n}.$$

This gives us 
$$\I[G] \le \sqrt{2/\pi}\sqrt{\chi(G)}\sqrt{n}+O(\chi(G))$$
proving the corollary.
\end{proof}
\begin{rmk}
This corollary, combined with the facts that planar graphs are 4-colorable and that a graph with maximum degree $\Delta$ has $\chi(G)\le \Delta +1$, we can prove that $\mathrm{I}[G] \le O(\sqrt{n})$ for any planar graphs or graphs with bounded degree.
\end{rmk}

We can improve the Covering Lemma by allowing a non-uniform probability distribution on the elements of the cover.

\begin{lem}[Randomized Covering Lemma]
Let $G =(V,E)$ be a graph with induced cover $C = \{G_{1}, \cdots, G_{k}\}$ and let $P$ be a probability distribution on covering $C$ such that $ \Pr_{G_{i} \sim P}[v \in G_{i}] \geq \epsilon, \forall v \in V. $ Then,
$$ \I[G] \leq \frac{1}{\epsilon} \cdot \expect_{G_{i} \sim P}[\I[G_{i}]].$$
\end{lem}

The proof follows the same approach as that of the Covering Lemma. This improved lemma allows us to change the chromatic number bound into one involving the fractional chromatic number of $G$.

\begin{defn} \label{fract}
Let $G=(V,E)$ be a graph and let $\mathcal{I}$ denote the collection of independent sets in $G$. The fractional chromatic number $\chi_{f}(G)$ is the optimal value of the linear program 

\begin{align*}
\mathrm{min} &\sum_{S \in \mathcal{I}} x_{S} \\
\mathrm{s.t.} &\sum_{S \ni v } x_{S} \geq 1, \forall v \in V
\end{align*}
\end{defn}

\subsection*{Proof of Theorem \ref{fracch}}

If we scale the optimal solution of the LP by $\chi_{f}(G)^{-1}$, we obtain a probability distribution $P $ on the independent sets of $G$ such that $\forall v \in V,$ $$ p_{v} := \Pr_{S \sim P}[v \in S] \geq \chi_{f}(G)^{-1}.$$ By the Randomized Covering Lemma,

$$ \I[f] \leq \chi_{f}(G) \cdot \expect_{S \sim P}[\I[S]]. $$

Since $S$ is an independent set of $G$, if $f \in Q_{S}$, then $f_{S|z}$ is an LTF for any $z \in \{-1,1\}^{|S^{c}|}$ and so
$$ \I[f] \leq \chi_{f}(G) \cdot \expect_{S \sim P}[\sqrt{|S|}] \leq \chi_{f}(G) \cdot \sqrt{\expect_{S \sim P}[|S|]}. $$

Moreover, let $q_{v}$ be the value in [0,1] such that $p_{v} \cdot (1-q_{v}) = \chi_{f}(G)^{-1}$. We construct a new probability distribution $\tilde{P}$ by, for every independent set $S$, removing each element in $S$ with probability $q_{v}$. Using distribution $\tilde{P}$, $$ \Pr_{S \sim \tilde{P}}[v \in S] = p_{v} \cdot (1-q_{v}) = \chi_{f}(G)^{-1},$$ 
and so $$ \expect_{S \sim \tilde{P}}[|S|] = \sum_{v \in V} \Pr_{S \sim \tilde{P}}[v \in S] = n \cdot \chi_{f}(G)^{-1}, $$ as desired.
\qed \\

To prove Theorem \ref{edgebound}, we'll borrow some simple ideas from \cite{Harsha}, which were used therein to prove an influence bound of $O(n^{1-\frac{1}{2^{d}}})$ for degree-$d$ PTFs.

\subsection*{Proof of Theorem \ref{edgebound}}
Let $f(x) = \sgn(p(x))$, where the quadratic polynomial $p(x) \neq 0$ for all $x \in \{-1,1\}^{n}$. For some fixed $x$, if $D_{i}f(x) > 0 $, then $ p(x^{i \rightarrow 1}) > p(x^{i \rightarrow -1}) $ and so $D_{i}f(x) = \sgn(D_{i}p(x))$ whenever the left side is nonzero. Hence,

$$\Inf_{i}[f] = \expect_{x \sim \{-1,1\}^{n}}[|D_{i}f|] \\
= \expect_{x}[\sgn(D_{i}p)\cdot D_{i}f].$$

Moreover, using the fact that $f = x_{i}\cdot D_{i}f + E_{i}$,
\begin{align*}
\expect_{x}[x_{i}\cdot \sgn(D_{i}p(x)) \cdot f] &= \expect_{x}[\sgn(D_{i}p(x)) \cdot D_{i}f] + \expect_{x}[x_{i} \cdot \sgn(D_{i}p(x)) \cdot E_{i}] \\ &= \expect_{x}[\sgn(D_{i}p(x)) \cdot D_{i}f] + \expect_{x}[x_{i}]\expect_{x}[\sgn(D_{i}p(x)) \cdot E_{i}],
\end{align*}
where the right term is zero as $E_{i}$ and $\sgn(D_{i}p(x)) $ are independent from $x_{i}$. Thus,

$$ \Inf_{i}[f] = \expect_{x}[x_{i} \cdot \sgn(D_{i}p(x))  \cdot f].$$
Using Cauchy-Schwarz,

\begin{align*}
\I[f] = \sum_{i \in [n]} \expect_{x}[x_{i} \cdot \sgn(D_{i}p) \cdot f] &\leq \expect_{x}[|\sum_{i\in [n]}x_{i} \cdot \sgn(D_{i}p)|]\\
&\leq \sqrt{n + \sum_{i \neq j}\expect_{x}[x_{i}x_{j} \cdot \sgn(D_{i}p) \cdot \sgn(D_{j}p)]}.
\end{align*}

For ease of notation, let $f_{i}(x) := \sgn(D_{i}p(x))$. Since $f_{i}$ does not depend on $x_{i}$, 

\begin{align*}
\expect_{x \sim \{-1,1\}^{n}}[x_{i}x_{j}f_{i}f_{j}] &= \expect_{x}\left[x_{j}f_{i}\left(\frac{f_{j}(x^{i \rightarrow 1})-f_{j}(x^{i \rightarrow -1})}{2} \right)\right] \\
&= \expect_{x}\left[x_{j}f_{i}\cdot D_{i}f_{j}\right].
\end{align*}

Similarly, since $D_{j}f_{i}$ does not depend on $x_{j}$,  we can repeat the same process with $f_{i}$ to obtain
$$ \expect_{x}[x_{i}x_{j}f_{i}f_{j}] = \expect_{x}\left[D_{j}f_{i}\cdot D_{i}f_{j}\right]. $$

Moreover, using Cauchy-Schwarz,

\begin{align*}
\expect_{x}[ D_{j}f_{i}\cdot D_{i}f_{j}] &\leq \sqrt{\expect_{x}[ D_{j}f_{i}^{2}]\expect_{x}[D_{i}f_{j}^{2}]} =\sqrt{\Inf_{j}[f_{i}]\cdot \Inf_{i}[f_{j}]}\\
&\leq \frac{\Inf_{j}[f_{i}]+ \Inf_{i}[f_{j}]}{2},
\end{align*}

where the last step follows by the AM-GM inequality. Therefore,

$$ \I[f] \leq \sqrt{n + \sum_{i\in [n]}\I[\sgn(D_{i}p)]}.$$
Since $p(x) = \sum_{i < j} a_{i,j} x_i x_j + \sum_{i}b_{i}x_{i} + c$, we have that, for a fixed $i$, 
 $$ \sgn(D_{i}p(x)) = \sgn\left(\sum_{j|(i,j) \in E)} a_{i,j} x_j + b_{i}\right),$$ so $ \sgn(D_{i}p(x)) $ is an LTF on the variables adjacent to $x_i$ in $G$ and therefore $$ \I[\sgn(D_{i}p(x))] \leq \sqrt{\deg(i)}.$$ Hence, by Cauchy-Schwarz, $$\I[f] \leq \sqrt{n + \sum_{i\in [n]} \sqrt{\deg(i)}} \leq \sqrt{n + \sqrt{2|E|n}}.$$ \qed 

\textbf{Remark:} Observe that by taking $G_1$ in the covering lemma to be the induced subgraph on any set of vertices we wish to ignore, we can apply Theorem 1 or Theorem 2 to $G \setminus G_1$ and obtain
$$\I[G] \leq \min\{\sqrt{n+ \sqrt{2|E(G\setminus G_1)|}}, \sqrt{\chi_f(G\setminus G_1)}\cdot \sqrt{n}\} + \I[G_1]$$
If, in particular, $|G_1| < n/\log^C(n)$ for large enough $C > 0$, then Kane's bound yields $\I[G_1] \leq \sqrt{n}$, and so in proving a function satisfies the weak Gotsman-Linial conjecture, one may throw away as many as $n/\log^C(n)$ variables. 

\section{Counterexamples to the Gotsman-Linial Conjecture} \label{counterexample}

For $n \leq 4$, the number of boolean functions on $n$ variables is $2^{2^n} \leq 2^{16}$, so an exhaustive search is tractable. In those cases, our exhaustive search verified that the Gotsman-Linial conjecture is true for $d =2, n \leq 4$.

For $n = 5$, we searched only through truth tables of functions symmetric in the last two coordinates (reducing the size of the search space from $2^{32}$ to $2^{24}$), first screening for high influence\footnote{Since an average function has influence
$\expect_{f}[\I[f]] = 2.5$, standard concentration bounds imply that most functions will fail to pass this screening.} ($>\I_{GL}(5,2) = 3.125$) before testing for QTF-ness. 

To test whether a boolean function $f$ can be represented as a QTF, we simply test for feasibility of the following LP:

\begin{align*}
&\text{find } q \in \R^{16} \\
&\text{such that } f_x (Tq)_x \ge 1, \, \, \forall x \in \{-1,1\}^5
\end{align*}

where the linear map $T: \R^{16} \rightarrow \R^{32}$ is the evaluation map on the space of multilinear quadratic functions on $\{-1,1\}^5$. If this LP is feasible, then any solution yields the coefficients of a quadratic polynomial $q$ such that at all of the inputs $x \in \{-1,1\}^5$, $\sgn(q(x)) = f(x)$.
This search produces the counterexample 
\begin{align*}
f(x)=\sgn(-7x_1x_2 -7x_1x_3 + x_1x_4 + x_1x_5 -7x_2x_3 + 2x_2x_4 + 2x_2x_5 + 2x_3x_4 \\+ 2x_3x_5 + 3x_4x_5 + x_1 + 2x_2 + 2x_3 + 3x_4 + 3x_5 -4)
\end{align*}
with influence $3.1875 > \I_{GL}(5,2).$

Observe that this QTF can be written as
$$\sgn(p(x_1, x_2 + x_3, x_4 + x_5))$$ for
\be \nn
p(x, y, z) := 2x(1 - 7y + z) + 4 y - 7y^2 + 4yz + 6z + 3z^2. 
\ee
On 7 variables, the function 
$$\sgn(p(x_1, x_2 + x_3 + x_4 + x_5, x_6 + x_7))$$
also has high influence: $249/64 > 245/64 = \I_{GL}(7, 2).$ In fact, we now prove that the functions obtained in this way:
$$f_n(x):=\sgn\left(p\left(x_1,\sum_{i=2}^{n-2} x_i, x_{n-1}+x_n\right)\right)$$
always have higher influence than $\I_{GL}(n,2)$ for odd $n \geq 5$.

\subsection*{Proof of Theorem \ref{generalize}}
Let $n \geq 9$ be odd and set $m= n-3$. By considering the 6 possible functions obtained from $f_n$ by restricting the values of $x_1$ and $x_{n-1} + x_n$, we can find the influence of the remaining $n-3$ coordinates directly by counting, and then appeal to Lemma \ref{infrestriction} to conclude that 
\begin{align*}
\sum_{i=2}^{n-2}\Inf_i[f_n] &= 2^{-m}\left[ \binom{m}{m/2}\cdot \frac{11m}{16} + \binom{m}{m/2 -1}\left(\frac{15m}{16}+ \frac{7}{8}\right) \right.\\ 
&\left. + \binom{m}{m/2 - 2}\left(\frac{5m}{16} + \frac{3}{4}\right) + \binom{m}{m/2 - 3}\left(\frac{m+6}{15}\right) \right]
\end{align*}

The influences of $x_1$, $x_{n-1}$ and $x_n$ can also be obtained directly:

$$\Inf_1[f_n] = 2^{-m}\left[ \binom{m}{m/2}\cdot \frac{3}{4} + \binom{m}{m/2 -1} \cdot \frac{5}{4} + \binom{m}{m/2 -2}\cdot \frac{1}{4} \right]$$

$$\Inf_{n-1}[f_n] = \Inf_{n}[f_n] = 2^{-m}\left[ \binom{m}{m/2}\cdot \frac{3}{4} + \binom{m}{m/2 -1} \cdot \frac{3}{4} + \binom{m}{m/2 -2}\cdot \frac{1}{4} \right]$$

Adding these expressions together yields

\begin{align*}
\I[f_n] = &2^{-m}\left[\binom{m}{m/2}\left(\frac{11m}{16} + \frac{9}{4}\right) + \binom{m}{m/2 -1}\left(\frac{15m}{16} + \frac{29}{8}\right) \right.\\ 
&\left.+ \binom{m}{m/2-2}\left(\frac{5m}{16} + \frac{3}{2}\right) + \binom{m}{m/2 - 3}\left(\frac{m}{16} + \frac{3}{8}\right)\right]
\end{align*}
while by Lemma \ref{real GL},
$$\I_{GL}(2, n) = 2^{-m}\binom{m+3}{m/2 + 1}\left(\frac{m+3}{4} \right).$$
Computing the ratio, one finds

$$\frac{\I[f_n]}{\I_{GL}(2,n)} = 1+ \frac{7}{32n} - \frac{3}{32(n-2)}+ \frac{3}{16 n^2} = 1 + \frac{1}{8n} + O\mathopen{}\left(\frac{1}{n^2}\right). $$ 
\qed\\

Thus, $\{f_n\}_{n\geq 5, \text{ odd}}$ provides an infinite family of counterexamples to the Gotsman-Linial conjecture for $d = 2$, beating the supposed upper bound by an additive $\Omega(1/\sqrt{n})$.

\subsection{Discussion}
We note that all of the counterexamples discovered have a complete support. This raises a natural question: can a graph with incomplete support still have influence greater than $\I_{GL}$? 

Let $H$ and $G$ be graphs, with $H \subseteq G$. Then, as we can change each coefficient of a PTF by a small amount without altering its truth table, we of course have $\I[h] \le \I[g]$. However, it is not obvious when equality should hold. Suppose $G_0 \subseteq G_1 \subseteq \cdots \subseteq G_m = G$ is a chain of subgraphs of $G$, starting from the empty graph $G_0$ and adding one edge at a time. If $G = K_n$, then since $\I[G] - \I[G_0] = \tilde{\Theta}(\sqrt{n})$, the average gap $\I[G_{i+1}] - \I[G_i]$ is $\tilde{\Theta}(n^{-3/2})$, where the $\tilde{\Theta}$ hides the possible polylogarithmic factors from Kane's upper bound. For a random sequence of edges, are the gaps likely to be evenly distributed, or does removing some small, special subset of edges account for most of the drop in influence? More generally, for non-isomorphic graphs $G$ and $G'$, what can be said about $\I[G]$ versus $\I[G']$?

In order to gain some intuition, we decided to explicitly find the maximum influence attained by a function with each support. As there are only $2^{2^4}$ boolean functions on four variables, we used the method outlined in the previous subsection to find these values. The results are shown in Table \ref{graphdata}.

We can see that the maximal influence of 3.0 was only attained by a function with full support. As this also appeared to be the case in our search for a counterexample in the case of odd $n$ for $d=2$, it seems reasonable to conjecture that for each $n$ and $d$, the influence maximizer will have full support. As a stronger conjecture, it seems feasible that any function violating Gotsman-Linial's influence bound must have full support. 

\begin{table}[!ht]
\centering\renewcommand{\arraystretch}{2}
\begin{tabu}{ccrccrccrcc}
\tabucline[2pt]{-}
$G$ & $\I[G]$ && $G$ & $\I[G]$&& $G$ & $\I[G]$ && $G$ & $\I[G]$\\
\tabucline[1pt]{-}
\quadgraph{(1)--(2) (3)--(4)  }&2&
&\quadgraph{(1)--(2) (1)--(4) (2)--(3) (3)--(4) }&2 &
&\quadgraph{(1)--(2) (1)--(3) (2)--(3) (3)--(4) }&2.5 &
&\quadgraph{(1)--(2) (1)--(3) (1)--(4) (2)--(3) (2)--(4) (3)--(4) }&3
 \\ 
\quadgraph{(1)--(2) (2)--(3) (3)--(4)}&2 &
&\quadgraph{(1)--(4) (2)--(4) (3)--(4)}&2.5&
&\quadgraph{(1)--(3) (1)--(4) (2)--(3) (2)--(4) (2)--(4) (3)--(4) }&2.5
&  & \\ \hline
\tabucline[2pt]{-}
\end{tabu}
\caption{Maximum influences of QTFs supported by graphs on 4 vertices}
\label{graphdata}
\end{table}

\section{Acknowledgements}
The authors would like to thank the organizers of SPUR+ program at MIT for funding this research. In particular, we thank Cris Negron and Henry Cohn for their guidance. 

\newpage
\appendix

\section{Fourier Analysis}\label{fourieranalysis}

\begin{prop}[Fourier Expansion] \label{fourier}
Let $f:\{-1,1\}^{n}\rightarrow \R$ be a real Boolean function. Then, there exists an unique representation as a multilinear polynomial
$$ f(x)= \sum_{S \subset [n]} \hat{f}(S) \cdot x_{S}, $$ where $x_{S} = \Pi_{i \in S}x_{i}.$
\end{prop}

\begin{proof}
Note that we can interpret the set of functions as a vector space in $\R^{2^{n}}$ as we can add functions point-wise and multiply by a scalar. By defining the inner-product, $$ \left<f,g\right> = \expect_{x \in \{-1,1\}^{n}}[f(x)\cdot g(x)], $$ the characters $x_{S}$ form an orthonormal basis.
\end{proof}

This representation is particularly useful as it allows us to ``read off'' properties of the function. In particular, the Fourier expansion allows us to show that the discrete derivative operator behaves like in the continuous case.

\begin{prop}
Let $f:\{-1,1\}^{n}\rightarrow \{-1,1\}$ be a Boolean function with Fourier expansion $ f(x) = \sum_{S \subset [n]} \hat{f}(S) \cdot x_{S} $. Then, 
$$ D_{i}f(x) = \sum_{S \ni i} \hat{f}(S) \cdot x_{S \backslash \{i\}}. $$
\end{prop}

\begin{proof}
Note that the derivative operator is linear, so it suffices to look at the characters $x_{S}$. In this case,
$$ D_{i}x_{S} = 
\begin{cases}
x_{S \backslash \{i\}}, & \text{if } i \in S \\
0, & \text{otherwise}
\end{cases}
$$ so the statement follows by Proposition \ref{fourier}.
\end{proof}

This result allows us to obtain an expression for the influence of a Boolean function in terms of Fourier coefficients.

\begin{prop}
Let $f:\{-1,1\}^{n}\rightarrow \{-1,1\}$ be a Boolean function with Fourier expansion $ f(x) = \sum_{S \subset [n]} \hat{f}(S) \cdot x_{S} $. Then, 
$$ \Inf_i[f] = \sum_{S \ni i} \hat{f}(S)^{2}. $$
\end{prop}

\begin{proof}
By Definition \ref{derivative}, note that 
$$ D_{i}f(x) = 
\begin{cases}
\pm 1, & \text{if } x_{i} \text{ is an influential coordinate} \\
0, & \text{otherwise}
\end{cases}
$$ so $(D_{i}f(x))^{2}$ is an indicator for the $i$th influence. Hence, 
$$ \Inf_i[f] = \expect_{x}[(D_{i}f(x))^{2}]= \sum_{S \ni i} \hat{f}(S)^{2}. $$ 
\end{proof}

Similar to the discrete derivative operator, we can define the $i$th expectation operator where we take the average of a function at fixed values of $x_{i}$.

\begin{defn}[Expectation Operator]
Let $f:\{-1,1\}^{n} \rightarrow \R$ be a real Boolean function. Then, the $i$th expectation operator maps $f$ to the function $E_{i}f:\{-1,1\}^{n-1}\rightarrow \set{-1,1}$ defined by
$$ E_{i}f(x) = \expect_{\bm{x_{i}} \sim \{-1,1\}} [f(x_{1}, \cdots, \bm{x_{i}}, \cdots, x_{n})] = \frac{f(x^{i\rightarrow 1})+f(x^{i\rightarrow -1})}{2}.$$
\end{defn}

\begin{prop}
Let $f:\{-1,1\}^{n}\rightarrow \{-1,1\}$ be a Boolean function with Fourier expansion $ f(x) = \sum_{S \subset [n]} \hat{f}(S) \cdot x_{S} $. Then, 
$$ E_{i}f(x) = \sum_{S|i\notin S} \hat{f}(S) \cdot x_{S}. $$
\end{prop}

The proof follows the same structure as Proposition \ref{derivative}. This result allows us to rewrite a function on $n$ variables in terms of those in at most $n-1$.

\begin{prop}
Let $f:\{-1,1\}^{n}\rightarrow \R$ be a real Boolean function. Then, 
$$ f(x) = x_{i} \cdot D_{i}f(x) + E_{i}f(x). $$
\end{prop}

\begin{thm}[Total influence of LTF's]
Let $f:\{-1,1\}^{n}\rightarrow \{-1,1\}$ be an LTF. Then, $\I[f] \leq \sqrt{n}$.
\end{thm}

\begin{proof}
(Our proof follows \cite{ABF}.) Let $f(x) = \sgn(a_{0}+\sum_{i\in [n]} a_{i}x_{i})$ with real coefficients. Note that for a fixed $i$,  $f(x^{i\rightarrow 1}) \geq f(x^{i\rightarrow -1})$ when $a_{i} > 0$ and $f(x^{i\rightarrow 1}) \leq f(x^{i\rightarrow -1})$ when $a_{i} < 0$. Hence, the discrete derivative $\pm D_{i}f(x)$ will be a 0-1 indicator for the $i$th influence, where the sign depends on which case applies. Hence,
$$ \Inf_{i}[f] = \expect_{x \sim \{-1,1\}^{n}}[\pm D_{i}f(x)] = |\hat{f}({i})| $$ so by Cauchy-Schwarz,
$$ \I[f] = \sum_{i \in [n]}|\hat{f}({i})| \leq \sqrt{n \cdot \sum_{i \in [n]} \hat{f}({i})^{2}} = \sqrt{n}, $$
where we use the fact that
$$ 1 = \expect_{x \sim \{-1,1\}^{n}}[f(x)^{2}] = \sum_{i \in [n]} \hat{f}({i})^{2}.$$
\end{proof}

\begin{lem}
 For the edgeless graph $G = (V, \emptyset)$ on $n$ vertices, $$\I[G] = \I[\MAJ_n] = (\sqrt{2/\pi}+ o_n(1))\sqrt{n}$$
\end{lem}
\begin{proof}
The set $Q_G$ is exactly the set of LTFs on $n$ variables. If $f$ is an LTF on $n$ variables, then following the approach of the previous lemma,
$$\I[f] = \sum_{i \in [n]} |\hat{f}(i)| = \expect_{x}[(x_1 + \cdots + x_n)f(x)] \leq \expect_{x}[|x_1 + \cdots + x_n|]$$
with equality if and only if $f(x) = \sgn(x_1 + \cdots x_n) = \MAJ_n(x)$. Thus, $\I[G] = \I[\MAJ_n]$. To get the asymptotic formula, we apply the central limit theorem to conclude
$$\frac{x_1 + \cdots + x_n}{\sqrt{n}} \overset{\text{d}}\longrightarrow \mathcal{N}(0,1)$$
and hence
$$\frac{1}{\sqrt{n}}\E[|x_1 + \cdots + x_n|] \to \E[|\mathcal{N}(0,1)|] = \sqrt{2/\pi}.$$
\end{proof}

\section{$\I_{GL}$ computation}\label{igl}
We show the derivation in the case of $d=2$ for the sake of simplicity, but the argument generalizes to all $d$. 

\begin{lem}\label{real GL}

Let $f:\{-1,1\}^{n}\rightarrow \{-1,1\}$ be a symmetric QTF. Then,
$$ \I[f] \leq
\begin{cases}
n \cdot 2^{-n+1} \binom{n}{\frac{n-1}{2}}, & n \text{ is odd} \\
n \cdot 2^{-n} \left[2\binom{n}{\frac{n-2}{2}} + \binom{n}{\frac{n}{2}}\right], & n \text{ is even}
\end{cases}
.$$
\end{lem}

\begin{proof}
Let $f$ be a symmetric QTF with the form $f(x) = \sgn(p(t))$, where $p$ is a polynomial on the sum of variables $t := \sum_{i \in [n]} x_{i}$. Ideally, we would like the roots of $p$ to be between the possible values of $t$ so that changing the value of a coordinate corresponds to a change in the sign of $p$. Moreover, since the number of inputs $x$ such that $t(x) = k$ is $\binom{n}{\frac{n-k}{2}}$ (a decreasing function of $k$), we would like  the roots to be close to 0. Hence,
\begin{enumerate}
\item When $n$ is odd, we choose the roots of $p$ to be at $t = 0,2.$ Thus, by the well-known binomial identity $ \binom{n}{k} = \binom{n-1}{k} + \binom{n-1}{k-1} $,
$$ \Inf_{1}[f] = 2^{-n} \left[\binom{n}{\frac{n-1}{2}} + \binom{n-1}{\frac{n-3}{2}} + \binom{n-1}{\frac{n-1}{2}}\right] = 2^{-n+1} \binom{n}{\frac{n-1}{2}} $$
where $x_{1}$ is restricted to be -1 at the boundary value $t = -1$ and, similarly, 1 at $t = 3.$
\item When $n$ is even, we choose the roots of $p$ to be at $t = -1,1.$ Similarly, $$ \Inf_{1}[f] = 2^{-n} \left[\binom{n}{\frac{n}{2}} + 2\binom{n-1}{\frac{n-2}{2}}\right],$$
where, as before, $x_{1}=-1$ when $t = -2$ and $x_{1}=1$ when $t = 2.$
Since $f$ is symmetric, $$\I[f] = n \cdot \Inf_{1}[f].$$
\end{enumerate}
\end{proof}

Note that, for $d=2$ and $n=5$, we obtain the value 3.125. We take care to calculate this value as there is ambiguity in the original conjecture that results in a higher value of 3.4375, which was verified through exhaustive search to be unattainable by symmetric QTFs.

\end{document}